\documentclass[lettersize,journal]{IEEEtran}
\usepackage{amsmath,amsfonts}
\usepackage{algorithm}
\usepackage{mathtools}
\usepackage{algorithmicx, algpseudocode}
\usepackage{amsthm}

\usepackage{array}
\usepackage{textcomp}
\usepackage{stfloats}
\usepackage{url}
\usepackage{verbatim}
\usepackage{graphicx}
\usepackage{cite}
\usepackage{enumitem}

\usepackage{multirow}
\usepackage{booktabs} 
\usepackage{subcaption}
\usepackage{bm}
\usepackage{xspace}

\newcommand{\ie}{\emph{i.e.},\xspace}
\newcommand{\eg}{\emph{e.g.},\xspace}

\newtheorem{theorem}{Theorem}
\newtheorem{lemma}[theorem]{Lemma}

\newcommand{\mmnm}{{LIRDRec}}

\begin{document}

\title{Learning Item Representations Directly from Multimodal Features for Effective Recommendation}

\author{Xin Zhou, Xiaoxiong Zhang, Dusit Niyato,~\IEEEmembership{Fellow,~IEEE,} and Zhiqi Shen,~\IEEEmembership{Senior Member,~IEEE,}

\thanks{Xin Zhou, Xiaoxiong Zhang, Dusit Niyato and Zhiqi Shen are with the College of Computing and Data Science, Nanyang Technological University, Singapore. 
(E-mail: \{xin.zhou, dniyato, zqshen\}@ntu.edu.sg; zhan0552@e.ntu.edu.sg.)}

}


\maketitle

\begin{abstract}
Conventional multimodal recommender systems predominantly leverage Bayesian Personalized Ranking (BPR) optimization to learn item representations by amalgamating item identity (ID) embeddings with multimodal features. Nevertheless, our empirical and theoretical findings unequivocally demonstrate a pronounced optimization gradient bias in favor of acquiring representations from multimodal features over item ID embeddings. As a consequence, item ID embeddings frequently exhibit suboptimal characteristics despite the convergence of multimodal feature parameters. Given the rich informational content inherent in multimodal features, in this paper, we propose a novel model (\ie~\mmnm{}) that learns item representations directly from these features to augment recommendation performance.
Nevertheless, directly leveraging multimodal features for recommendation is a non-trivial endeavor. Recognizing that features derived from each modality may capture disparate yet correlated aspects of items, we propose a multimodal transformation mechanism, integrated with modality-specific encoders, to effectively fuse features from all modalities. This enables the explicit capture of shared information across different modalities, thereby enhancing recommendation efficacy. Moreover, to differentiate the influence of diverse modality types, we devise a progressive weight copying fusion module within \mmnm{}. This module incrementally learns the weight assigned to each modality in synthesizing the final user or item representations. 
Finally, we utilize the powerful visual understanding of Multimodal Large Language Models (MLLMs) to convert the item images into texts and extract semantics embeddings upon the texts via Large Lanuage Models (LLMs).
Empirical evaluations conducted on five real-world datasets validate the superiority of our approach relative to competing baselines. It is worth noting the proposed model, equipped with embeddings extracted from MLLMs and LLMs, can further improve the recommendation accuracy of NDCG@20 by an average of 4.21\% compared to the original embeddings.
Moreover, as our model directly extracts item representations from multimodal features, it exhibits accelerated startup recommendation performance, surpassing state-of-the-art methods in cold-start scenarios. Code available at \url{https://github.com/enoche/LIRDRec}.
\end{abstract}


\section{Introduction}
\IEEEPARstart{T}{he} ubiquity of multimedia content on the internet and diverse platforms has led to the proliferation of multimodal data, encompassing textual modalities such as product descriptions, as well as visual modalities like images and videos~\cite{baltruvsaitis2018multimodal}. Consequently, multimodal recommender systems (MMRSs) have emerged as a vital research area, harnessing information from multiple modalities, including text, images, and audio, to deliver personalized recommendations to users~\cite{deldjoo2020recommender,salah2020cornac,zhou2023comprehensive,liu2024multimodal,liu2024multimodalKDD}.
Prior studies~\cite{LATTICE2021MM, FREEDOM2023MM,zhou2023enhancing} have demonstrated that MMRSs can effectively augment the user experience by providing more accurate recommendations tailored to individual preferences. Additionally, MMRSs have the potential to mitigate the cold-start problem, which arises when new items or users lack sufficient historical data for reliable recommendations.

The majority of existing Multimodal Recommender Systems (MMRSs) learn item representations by fusing item identity (ID) embeddings with corresponding multimodal features~\cite{VBPR2016AAAI,MGCN2023MM,FREEDOM2023MM}. For instance, VBPR~\cite{VBPR2016AAAI} and GRCN~\cite{wei2020graph} directly concatenate ID embeddings with transformed multimodal features to form item representations. Alternatively, methods such as MGCN~\cite{MGCN2023MM}, TMFUN~\cite{zhou2023attention} and LGMRec~\cite{guo2024lgmrec} augment ID embeddings with multimodal features and employ additive fusion. FREEDOM~\cite{FREEDOM2023MM} adopts a late fusion strategy by summing losses from ID embeddings and multimodal features. While these approaches have demonstrated performance improvements, they fundamentally treat multimodal information as auxiliary to the primary task of learning ID embeddings.

\begin{figure}[bpt]
  \centering
  \includegraphics[width=\linewidth, trim={10 10 10 10},clip]{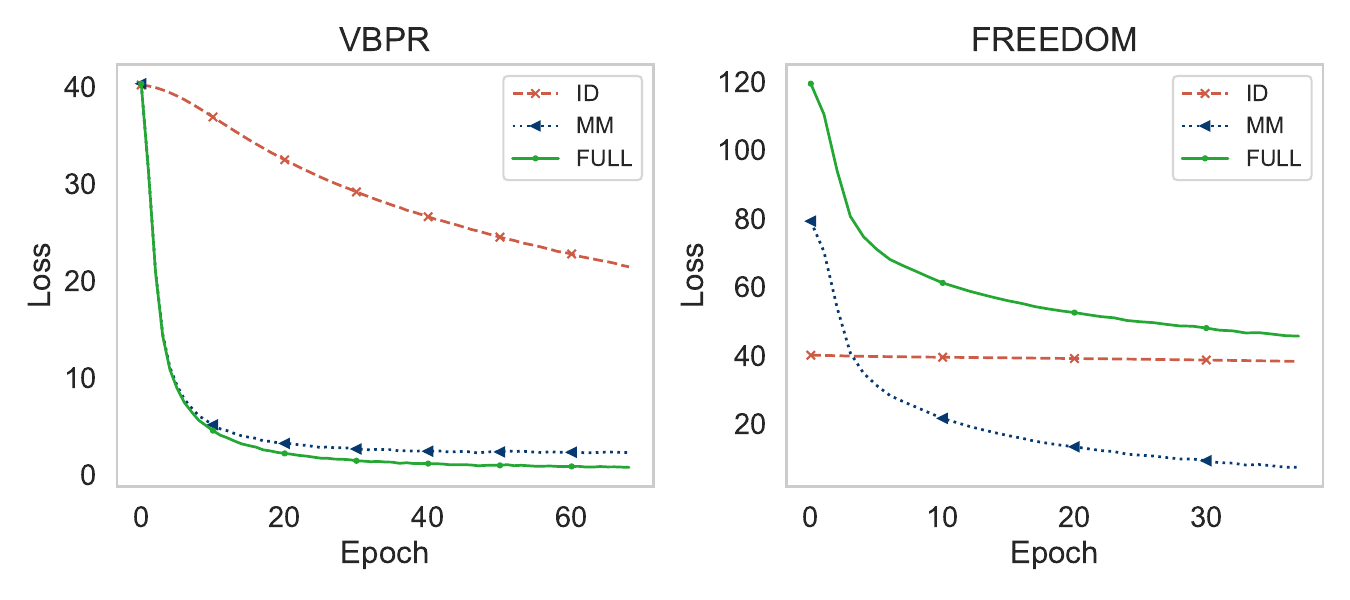}
  \caption[]{Comparison of training losses between ID embedding and multimodal feature (MM) learning in VBPR~\cite{VBPR2016AAAI} and FREEDOM~\cite{FREEDOM2023MM}. The magnitude of the gradient flow in MM learning exhibits a steeper and more pronounced decay compared to the gradients observed in ID embedding optimization.}
  \label{fig:intro_loss}
  \vspace{-10pt}
\end{figure}

To demonstrate the potential detrimental impact of the aforementioned learning paradigm on ID embedding acquisition, we conduct an experimental analysis on two representative multimodal recommender systems: VBPR~\cite{VBPR2016AAAI} and FREEDOM~\cite{FREEDOM2023MM}. Specifically, we decompose the Bayesian Personalized Ranking (BPR) loss~\cite{rendle2009bpr} into components attributed to ID embeddings and multimodal features, respectively, as calculated on the Baby dataset (Table~\ref{tab:datasets}) at each epoch until early stopping. The results are visualized in Fig.~\ref{fig:intro_loss}, which also depicts the total BPR loss.  It is important to note that VBPR employs early fusion, concatenating ID embeddings with transformed multimodal features, resulting in a non-linear relationship between the total loss and its components. 
Our analysis reveals a more rapid decline in the gradient of the multimodal feature component than that of the ID embedding component. Moreover, the gradient of the ID embedding component persists at a notably slower rate even after the multimodal component appears to converge. 
As a result, the optimization of the networks of learning multimodal features is often adversely affected by the optimization of ID embeddings, as its learning loss decline more rapidly than that of ID embeddings(Fig. 1). Consequently, the potential expressiveness of multimodal features is not fully exploited in these existing frameworks. This observation motivated us to develop \mmnm{}, a framework that learns item representations directly from multimodal features without relying on ID embeddings.
Furthermore, our theoretical analysis demonstrates that, under equivalent loss conditions, the gradient magnitude associated with learning ID embeddings is smaller than that associated with learning multimodal features during the early stages of training. In contrast to trainable ID embeddings, typically initialized randomly~\cite{wei2019mmgcn}, multimodal features inherently encode relevant item information. As a result, models incorporating multimodal features exhibit greater information gain compared to those relying solely on ID embeddings.
This study investigates whether recommendation systems that learn item representations exclusively from multimodal features can outperform their ID-inclusive counterparts.

While several studies have explored non-ID recommendation~\cite{hou2022towards,hou2023learning,zhang2024id,qu2024elephant,zhang2024dual}, they primarily focus on sequential recommendation, which necessitates item interaction sequences for model training. Existing research relying solely on multimodal features~\cite{yuan2023go} has yielded inferior results compared to ID-based recommendation approaches. 
Consequently, to fully exploit the potential of multimodal features, we propose \mmnm{}, a model that \underline{L}earns \underline{I}tem \underline{R}epresentations \underline{D}irectly from multimodal features for \underline{Rec}ommendation.
Specifically, \mmnm{} employs a dual-pronged strategy to harness multimodal features: \textit{i)} enhancing uni-modal features by applying the 2-D Discrete Cosine Transform (DCT) and subsequently synthesize the interplay among heterogeneous feature modalities through a set of designed transformations, and \textit{ii)} discriminating modality-specific information via progressive weight copying (PWC).
Differentiating between relevant and irrelevant information within multimodal features is crucial, as these features, often extracted from pre-trained encoders, can contain noise or task-irrelevant components~\cite{liu2023semantic}.
To address this, the PWC module partitions the intermediate representation into segments and employs neural networks to assign weights to each segment. These weights are exponentially updated during training, facilitating a gradual and smooth adjustment of the model's focus on different feature aspects.
Our contributions can be summarized as follows:
\begin{itemize} 
\itemsep0em 
\item We conduct a comprehensive empirical and theoretical analysis of the training loss associated with item ID embeddings and multimodal features in recommendation model optimization. Our analysis reveals a pronounced optimization gradient bias in favor of learning item representations from multimodal features.
\item To exploit this bias, we propose \mmnm{}, a model exclusively focused on learning item representations from multimodal features for enhanced recommendation performance. \mmnm{} is designed to extract informative features from multimodal data and effectively discriminate between them using progressive weight copying.
\item We evaluate the proposed model on five public datasets, including a large-scale benchmark. Experimental results validate the efficacy of \mmnm{} in recommendation. Moreover, through comprehensive testing in various scenarios, \mmnm{} demonstrates rapid learning capabilities in directly acquiring item representations from multimodal features.
\end{itemize}

\begin{figure*}[ht]
  \centering
  \includegraphics[width=0.93\linewidth, trim={10 15 15 20},clip]{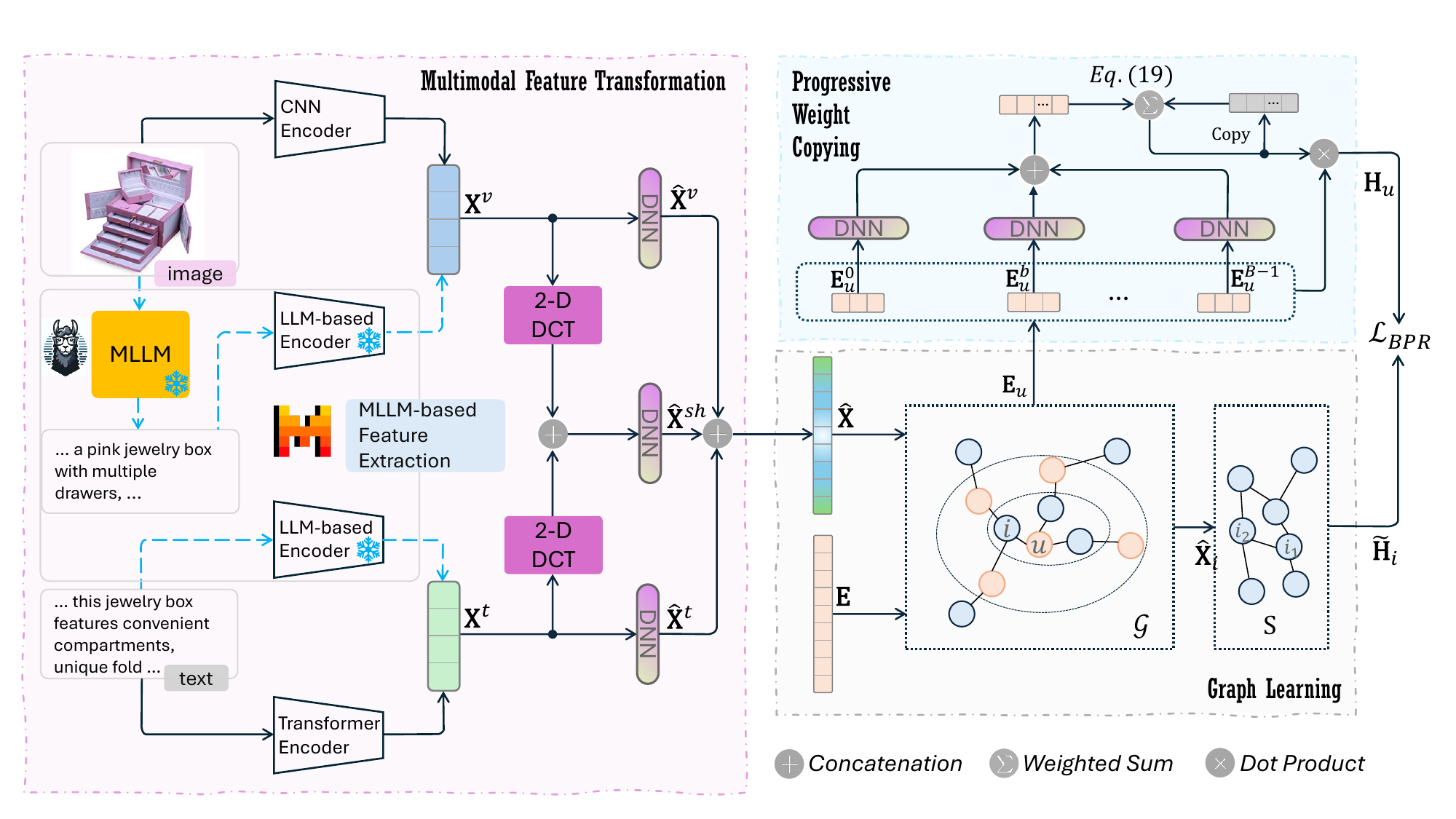}
  \caption[]{An overview of \mmnm{}. Instead of utilizing ID embeddings to represent items, \mmnm{} directly derives item representations from multimodal features.}
  \vspace{-16pt}
  \label{fig:model}
\end{figure*}

\section{The Preference for Multimodal Features in Recommendation: A Theoretical Analysis}
\label{sec:theo}
\subsection{Notations}
Consider a dataset represented as $\mathcal{D}=\{\mathcal{G}, \mathbf{X}^0, \ldots, \mathbf{X}^{|\mathcal{M}|-1}\}$, where $\mathcal{G} = (\mathcal{V},\mathcal{E})$ denotes the interaction bipartite graph. The sets $\mathcal{E}$ and $\mathcal{V}$ correspond to the edges and nodes of users and items, respectively. An edge $\mathcal{E}_{ui}=1$ within the graph $\mathcal{G}$ indicates an interaction between a user $u$ and an item $i$. The node set $\mathcal{V}=\mathcal{U} \cup \mathcal{I}$ encompasses all users and items, with $u\in \mathcal{U}$ and $i\in \mathcal{I}$ representing the user and item, respectively. $\mathbf{X}^m$ signifies the $m$-th modality feature matrix with dimension $\mathbb{R}^{|\mathcal{I}| \times d_m}$ for items, and $\mathcal{M}$ constitutes the set of modalities.
Our investigation primarily concentrates on two modalities, namely vision and text. The raw modality information for item $i$ is denoted as $\mathbf{x}^v_i$ and $\mathbf{x}^t_i$, respectively.
Based on this dataset $\mathcal{D}$, the objective of a multimodal recommender system is to recommend a ranked list of items to each user based on their preference score, which is defined as: $\hat{y}_{ui} = f_{\Theta}(u, i, \mathbf{x}^v_i, \mathbf{x}^t_i)$. $\Theta$ are the trainable parameters of the model $f$.
To optimize the model, MMRSs typically employ a pairwise Bayesian Personalized Ranking (BPR) loss~\cite{rendle2009bpr}, assigning higher predictive scores to observed user-item interactions compared to unobserved ones. That is,
\begin{equation}
	\mathcal{L}_{BPR} = \frac{1}{|\mathcal{R}|} \sum_{(u,i,j)\in \mathcal{R}} -\ln \sigma(\hat{y}_{ui} - \hat{y}_{uj}),
	\label{eq:loss}
\end{equation}
where $\mathcal{R}$ is the set of training instances, and each triple $(u,i,j)$ satisfies $\mathcal{E}_{ui} = 1 ~\text{and}~ \mathcal{E}_{uj} = 0$. $\sigma(\cdot)$ is the sigmoid function.

Conventional MMRSs~\cite{VBPR2016AAAI,FREEDOM2023MM} typically learn item representations by combining item ID embeddings with multimodal features. The multimodal features are initially projected into the same embedding space as the ID embeddings using a multilayer perceptron (MLP) before being fused at either an early or late stage. We theoretically demonstrate a bias towards multimodal features in these MMRSs.
\subsection{Theoretical Analysis}
\begin{theorem}
(\textbf{Gradient Magnitude Disparity in Early Training}): Let $\nabla \theta_{ID}$ and $\nabla \theta_{MM}$ denote the gradients of the loss function with respect to item ID embeddings and multimodal features, respectively. During the initial training epoch, the following inequality holds:

$$|\nabla \theta_{ID}| < |\nabla \theta_{MM}|,$$
where $|\cdot|$ represents the Euclidean norm.
\label{th:gradient}
\end{theorem}
\begin{proof}
    Taking VBPR~\cite{VBPR2016AAAI} as an example, the representations of a user and an item are defined as follows:
    \begin{equation}
	\mathbf{x}_u = \mathbf{x}^{ID}_u \oplus \mathbf{x}^{p}_u; \quad \mathbf{x}_i = \mathbf{x}^{ID}_i \oplus \mathrm{MLP}(\mathbf{x}^{v}_i \oplus \mathbf{x}^{t}_i), 
\end{equation}
where $\mathbf{x}^{ID}$ represents the user or item ID embedding, $\mathbf{x}^{p}_u$ denotes the user preference on multimodal features, and $\oplus$ symbolizes the concatenation of two vectors. The transformed representation from multimodal features is defined as $\mathbf{x}_i^{MM} = \mathbf{W}(\mathbf{x}^{v}_i \oplus \mathbf{x}^{t}_i)$, where $\mathbf{W}$ is a linear projection matrix. The preference score between user $u$ and item $i$ can be reformulated as:
\begin{equation}
\hat{y}_{ui} = \mathbf{x}_u^\top \mathbf{x}_i = (\mathbf{x}_u^{ID})^\top \mathbf{x}_i^{ID} + \mathbf{p}_u^\top \mathbf{x}_i^{MM}.
\end{equation}
Consequently, the BPR loss can be reformulated as:
\begin{equation}
\begin{aligned}
\mathcal{L}_{BPR} & = {\mathbb{E}} [-\ln \sigma\left(\hat{y}_{u i}-\hat{y}_{u j}\right)] \\
& = {\mathbb{E}} [-\ln \sigma\left((\mathbf{x}_u^{ID})^\top \mathbf{x}_i^{ID} + \mathbf{p}_u^\top \mathbf{x}_i^{MM} - \right.\\
& \quad \left. ((\mathbf{x}_u^{ID})^\top \mathbf{x}_j^{ID} + \mathbf{p}_u^\top \mathbf{x}_j^{MM})\right)]\\
&= {\mathbb{E}} [-\ln \sigma\left(z\right)],
\end{aligned}
\end{equation}
where $z=(\mathbf{x}_u^{ID})^\top \mathbf{x}_i^{ID} + \mathbf{p}_u^\top \mathbf{x}_i^{MM}-((\mathbf{x}_u^{ID})^\top \mathbf{x}_j^{ID} + \mathbf{p}_u^\top \mathbf{x}_j^{MM})$.
To isolate the impact of item representation learning, we compare the contributions of item ID and multimodal features.

The gradient magnitude with respect to $\mathbf{x}_i^{ID}$ is computed as follows:
\begin{equation}
\frac{\partial L_{BPR}}{\partial \mathbf{x}_i^{ID}}=(\sigma(z)-1) \cdot \mathbf{x}_u^{ID}.
\end{equation}
The gradient magnitude of $\mathbf{x}_i^{MM}$ can be derived as follows:
\begin{equation}
\frac{\partial L_{BPR}}{\partial \mathbf{x}_i^{MM}} =(\sigma(z)-1) \cdot 
 \mathbf{p}_u^\top.
\end{equation}
As both $\mathbf{x}_u^{ID}$ and $\mathbf{p}_u$ are initialized using the same method, they exhibit no significant differences in the initial stages of training. That is $\frac{\partial L_{BPR}}{\partial \mathbf{x}_i^{ID}} \approx \frac{\partial L_{BPR}}{\partial \mathbf{x}_i^{MM}}$. However, as $\mathbf{x}_i^{MM}$ is derived from $\mathbf{W}$, a matrix shared across all items, the gradient magnitude propagated to $\mathbf{W}$ during the transformation of $\mathbf{x}_i^{MM}$ in one epoch is $\nabla \theta_{MM} = |\mathcal{E}| \cdot \frac{\partial L_{BPR}}{\partial \mathbf{x}_i^{MM}}$.
The magnitude of the gradient with respect to ID embeddings during a single epoch can be expressed as $\nabla \theta_{ID} = |\mathcal{N}_i| \cdot \frac{\partial L_{BPR}}{\partial \mathbf{x}_i^{ID}}$, where $|\mathcal{N}_i|$ represents the degree of node $i$ in graph $\mathcal{G}$ (\ie the cardinality of its neighbor set). Given the inherent sparsity of graph $\mathcal{G}$, characterized by $|\mathcal{N}_i| \ll |\mathcal{G}|$, we can establish the following inequality: $|\nabla \theta_{ID}| < |\nabla \theta_{MM}|$.
\end{proof}

The aforementioned theorem elucidates why randomly initialized ID embeddings can potentially hinder learning progress in Fig.~\ref{fig:intro_loss}.

\begin{lemma}
Given the heightened informational relevance of multimodal features, models constructed upon such representations exhibit superior information gain.
\end{lemma}
\begin{proof}
First, recall that the Kullback-Leibler (KL) divergence between the posterior $P(\Theta|\mathcal{D})$ and the prior $P(\Theta)$ is defined as: 
\begin{equation}
KL(P(\Theta \mid \mathcal{D}) \parallel P(\Theta)) =  \int P(\Theta \mid \mathcal{D}) \ln\left(\frac{P(\Theta \mid \mathcal{D})}{P(\Theta)}\right) d\Theta.
\label{eq:kl}
\end{equation}
Applying Bayes' Theorem to Eq.~\eqref{eq:kl}, we have:
\begin{equation}
\begin{aligned}
KL(P(\Theta \mid \mathcal{D}) \parallel P(\Theta)) &= \\ \int P(\Theta \mid \mathcal{D}) \cdot & \ln \left(\frac{P(\mathcal{D} \mid \Theta) P(\Theta) / P(\mathcal{D})} {P(\Theta)}\right) d\Theta.
\end{aligned}
\label{eq:klal}
\end{equation}
The above equation~\eqref{eq:klal} can be simplified as:
\begin{equation}
\begin{aligned}
KL(P(\Theta \mid \mathcal{D}) \parallel P(\Theta)) & = \\ \int P(\Theta \mid \mathcal{D}) \cdot
& \left[\ln(P(\mathcal{D} \mid \Theta)) - \ln(P(\mathcal{D}))\right] d\Theta.
\label{eq:klsim}
\end{aligned}
\end{equation}
Recall that the expectation of the log-likelihood function with respect to the posterior distribution of the parameters can be represented as:
$
\mathbb{E}[\ln P(\mathcal{D} \mid \Theta)] = \int P(\Theta \mid \mathcal{D}) \log P(\mathcal{D} \mid \Theta)  d\Theta
$. We can rewrite Eq.~\eqref{eq:klsim} as:
\begin{equation}
KL(P(\Theta \mid \mathcal{D}) \parallel P(\Theta)) = \mathbb{E}[\log P(\mathcal{D} \mid \Theta)] - \log P(\mathcal{D}).
\end{equation}
KL divergence quantifies the information gleaned about parameter $\Theta$ from observed data $\mathcal{D}$, effectively measuring the divergence between posterior and prior distributions. Given the constancy of $\log P(\mathcal{D})$ for a fixed dataset, maximizing the expected log-likelihood is equivalent to maximizing KL divergence (information gain) up to an additive constant. In contrast to models directly mapping multimodal features to item representations, randomly initialized item ID embeddings introduce minimal supplementary information about items, resulting in a lower expected log-likelihood, $\mathbb{E}[\log P(\mathcal{D} \mid \Theta)]$. 
\end{proof}

\section{The Proposed Method}
\label{sec:model}
In our framework, we harness multimodal information in a dual manner: \textit{i)} by fabricating a latent item-item graph; \textit{ii)} by deriving a low-dimensional representation for items. 
We adopt the specifics of constructing the item-item graph, as we adhere strictly to the methodology delineated in~\cite{FREEDOM2023MM} for its construction, denoted as $\mathbf{S}$. 

\subsection{Multimodal Feature Transformation (MFT)}
\textit{\textbf{Modality-specific Feature Projection}.}
We utilize a Deep Neural Network (DNN) to harness the modality-specific relations by projecting each uni-modal feature into a shared low-dimensional space. Specifically, given a uni-modal feature matrix of items, denoted as $\mathbf{X}^m$, we derive the latent uni-modal representation as follows:
\begin{equation}
    \hat{\mathbf{X}}^m = \phi({\mathbf{X}^m \mathbf{W}_1^m  + \mathbf{b}_1^m}) \mathbf{W}^m_2,
    \label{eq:mm}
\end{equation}
where `$\phi(\cdot)$' is the Leaky ReLU activation function~\cite{maas2013rectifier}, and $\mathbf{W}_1^m \in \mathbb{R}^{d_m \times d_1}$, $\mathbf{W}^m_2 \in \mathbb{R}^{d_1 \times d}$, $\mathbf{b}_1^m \in  \mathbb{R}^{d_1}$ represent trainable weight matrices and bias vector, respectively. Here, $d_1$ and $d$ denote vector dimensions.

\textit{\textbf{2-D Discrete Cosine Transform (DCT)}.}
To exploit the comprehensive contextual information within individual modalities and their cross-modal relationships, we apply DCT~\cite{ahmed1974discrete} to decorrelate modality-specific features. These decorrelated features are then concatenated across all modalities to generate a unified latent representation of the item.

The DCT divides a signal into a sum of cosine waves at different frequencies, and has been extensively applied in deep learning-based image processing, demonstrating enhanced reasoning capabilities~\cite{gueguen2018faster, li2023discrete, su2024dctvit}. Its effectiveness mainly stems from two key properties: energy compaction, which concentrates signal energy into a small number of coefficients, and decorrelation, which ensures statistical independence between coefficients~\cite{scribano2023dct}.
Given a two-dimensional ($|\mathcal{I}| \times d_m$) modality-specific feature $\mathbf{X}^m$, we adopt the type-II DCT as:  
\begin{equation}
\small
T(\mathbf{X}^m)=\alpha_\mu \alpha_v \sum_{i=0}^{|\mathcal{I}|-1} \sum_{j=0}^{d_m-1} f^m_{(i, j)} \cos \frac{\pi(2 i+1) \mu}{2 |\mathcal{I}|} \cos \frac{\pi(2 j+1) v}{2 d_m},
\end{equation}
where $\alpha_x=\left\{\begin{array}{ll}
\sqrt{1 / N}, & x=0 \\
\sqrt{2 / N}, & 1 \leq x \leq N-1
\end{array}\right.$, $N=|\mathcal{I}|-1$ when x=$\mu$, and $N=d_m$ when x=$v$. It is noteworthy that the DCT is applied as a pre-processing step prior to model training. Consequently, the computational overhead associated with the DCT does not impact the training process.

Having transformed each modality's features, we concatenate them into a single feature vector. A subsequent DNN layer is responsible for extracting and exploiting the complex cross-modal relationships embedded within this concatenated representation. The resultant shared representation is denoted as:
\begin{equation}
    \hat{\mathbf{X}}^{sh} = \phi\Bigl( \bigl(T(\mathbf{X}^{0})  \oplus  \cdots  \oplus T(\mathbf{X}^{|\mathcal{M}|-1}) \bigr) \mathbf{W}_1^{sh}  + \mathbf{b}_1^{sh} \Bigr) \mathbf{W}^{sh}_2.
    \label{eq:sh}
\end{equation}
The dimensions of matrices $\mathbf{W}_1^{sh}$, $\mathbf{W}^{sh}_2$ and $\mathbf{b}_1^{sh}$ are $\mathbb{R}^{(\sum_{m=0}^{|\mathcal{M}|} d_m) \times d_1}$,  $\mathbb{R}^{d_1 \times d}$ and $\mathbb{R}^{d_1}$, respectively.

\textit{\textbf{Latent Multimodal Feature as Latent Item  Representation}.}
With the aforementioned transformations, the latent multimodal feature space encompasses  $|\mathcal{M}| + 1$ latent modalities of features, including  $\hat{\mathbf{X}}^{sh}$ and the modal-specific features. These latent representations are subsequently propagated through Graph Convolutional Networks (GCNs) to enhance the high-order user/item representations.

In contrast to conventional multimodal models~\cite{FREEDOM2023MM, zhang2023latent, zhou2024disentangled, guo2024lgmrec} that rely on initialized item ID embeddings, we construct item representations by concatenating the $|\mathcal{M}| + 1$ latent features into a unified feature vector:
\begin{equation}
    \widetilde{\mathbf{X}} = \hat{\mathbf{X}}^{0}  \oplus \cdots  \oplus \hat{\mathbf{X}}^{|\mathcal{M}|-1}  \oplus \hat{\mathbf{X}}^{sh}.
    \label{eq:itemrep}
\end{equation}
The dimension of $\widetilde{\mathbf{X}}$ is $\mathbb{R}^{|\mathcal{I}| \times \ell}$, where $\ell =|\mathcal{M}|d+d$.

\subsection{Graph Learning}
To accommodate user preference and attend to both the modality-specific latent feature and the multi-modality shared feature, we formulate a user ID embedding matrix, represented as $\mathbf{E} \in \mathbb{R}^{|\mathcal{U}| \times \ell}$.
For the propagation of information within the convolutional network, we employ LightGCN~\cite{he2020lightgcn}. 
We will simplify the presentation of graph learning in this section, as these concepts are widely adopted in existing multimodal models~\cite{LATTICE2021MM, FREEDOM2023MM}.
Specifically, the representations of user $u$ and item $i$ at the $(l+1)$-th graph convolution layer of $\mathcal{G}$ are derived as follows:
\begin{equation}
	\begin{split}
        \mathbf{E}_{u}^{(l+1)} &= \sum_{i\in \mathcal{N}_{u} } \frac{1}{\sqrt{\left | \mathcal{N}_{u}  \right | }\sqrt{\left | \mathcal{N}_{i}  \right | } } \widetilde{\mathbf{X}}_{i}^{(l)};\\ 
        \widetilde{\mathbf{X}}_{i}^{(l+1)} &= \sum_{u\in \mathcal{N}_{i} } \frac{1}{\sqrt{\left | \mathcal{N}_{u}  \right | }\sqrt{\left | \mathcal{N}_{i}  \right | } } \mathbf{E}_{u}^{(l)},
        \end{split}
\end{equation}
where $\mathcal{N}_u$ and $\mathcal{N}_i$ denote the set of first hop neighbors of $u$ and $i$ in $\mathcal{G}$, respectively.
Employing $L_{ui}$ layers of convolutional operations, we extract all representations from the hidden layers to formulate the final representations for both the user and the item:
\begin{equation}
	\begin{split}
		\mathbf{E}_u &= \mathrm{R{\scriptsize EADOUT}}(\mathbf{E}_u^0, \mathbf{E}_u^1, \cdots ,\mathbf{E}_u^{L_{ui}}); \\
		\widetilde{\mathbf{X}}_i &= \mathrm{R{\scriptsize EADOUT}}(\widetilde{\mathbf{X}}_i^0, \widetilde{\mathbf{X}}_i^1, \cdots, \widetilde{\mathbf{X}}_i^{L_{ui}}),
		\label{eq:lgn_layer_update}
	\end{split}
\end{equation}
where the $\mathrm{R{\scriptsize EADOUT}}$ function can be any differentiable function. We use the sum function to derive the representations.

Finally, we execute graph convolutions on the items with the item-item graph $\mathbf{S}$, to derive the ultimate item representation. The detailed procedures involved in this process are not elaborated upon here, as they are analogous to those used in the user-item graph~\cite{zhang2022latent, FREEDOM2023MM}.
With the last layer's representation $\widetilde{\mathbf{X}}_{i}^{L_{ii}}$ ($L_{ii}$ is the number of layers), we establish a residual connection with the initial item representation ($\widetilde{\mathbf{X}}_{i}^{0}$) to obtain the item final representation:
\begin{equation}
	\widetilde{\mathbf{H}}_i = \widetilde{\mathbf{X}}_{i}^{L_{ii}} + \widetilde{\mathbf{X}}_{i}^{0}.
	\label{eq:ii_emb}
\end{equation}
Performing convolutional operations on item multimodal feature offers two key benefits: \textit{i)} Users can immediately obtain informative representations by attending to multimodal information; 
\textit{ii)} It allows for the retention of both modality-specific and multimodal shared features, which are crucial for late fusion and recommendation tasks.

\subsection{Progressive Weight Copying (PWC)}
Given a user representation $\mathbf{E}_u$, we first separate it into $B$ chunks, where $B$ is typically set to $|\mathcal{M}| + 1$. For each chunk $\mathbf{E}_u^b$ ($b \in [B]$), we feed it into a DNN to obtain the weight value:
\begin{equation}
    {a}_b = \phi({\mathbf{E}_u^b \mathbf{W}_1^b  + \mathbf{b}_1^b}) \mathbf{W}^b_2 + \mathbf{b}_2^b,
    \label{eq:PWC}
\end{equation}
where $\mathbf{W}_1^b \in \mathbb{R}^{({\ell}//{B}) \times d_h}$, $\mathbf{W}^b_2 \in \mathbb{R}^{d_h \times 1}$, $\mathbf{b}_1^b \in  \mathbb{R}^{d_h}$ and $\mathbf{b}_2^b \in  \mathbb{R}^1$ are trainable weight matrices and bias vectors. $d_h$ is latent dimension.

To preserve the values produced by DNNs, we define a target network with more tightened and compact parameters.
The target network, parameterized as $\mathbf{\theta}$, contributes to the final weight value in an exponential fashion. Specifically, we define a decay base $\gamma$ and an initial decay rate $\tau$, both of which reside within the interval $(0, 1)$. 
Upon the completion of the $n$-{th} training iteration, we perform a parameter update according to the decay function specified in equation~\eqref{eq:decay}.
This strategy ensures a dynamic contribution of the target network to the final weight value, thereby enhancing the efficiency of the learning process.
\begin{equation}
\begin{aligned}
    &\eta = \gamma^n; \\
    &a_b = \frac{\tau \cdot \eta}{\tau \cdot \eta + 1 - \tau} \cdot a_b + \frac{1 - \tau}{\tau \cdot \eta + 1 - \tau} \cdot \theta_b;\\
    &\tau = \tau \cdot \eta;\quad \theta_b = a_b.
    \label{eq:decay}
\end{aligned}
\end{equation}
Utilizing the weights values $[a_b]$ derived from each training iteration, we assign weights to signify the importance of each chunk $\mathbf{E}_u^b$. Subsequently, we concatenate all chunks to form comprehensive user representations $\mathbf{H}_u$. This approach ensures a dynamic and efficient representation of users, taking into account the varying importance of different information chunks.
\begin{equation}
    \mathbf{H}_u = a_0  \mathbf{E}_u^0 \oplus   a_1  \mathbf{E}_u^1 \oplus \cdots \oplus a_{B-1} \mathbf{E}_u^{B-1}.
    \label{eq:attfl}
\end{equation}

\textit{\textbf{Relationship between PWC and Attention~\cite{vaswani2017attention}}.}
We systematically analyze the key differences and similarities between PWC and traditional attention mechanisms:
\emph{i}). \textbf{Difference:} While attention constructs Query/Key/Value with three matrices, PWC learns the weights of each block using a DNN.
\emph{ii}). \textbf{Difference:} PWC does not require dot product calculations to determine importance between elements, unlike attention's operation.
\emph{iii}). Similarity: Both PWC and attention use softmax for scaling scores.
\emph{iv}). \textbf{Difference:} Uniquely, PWC copies the calculated weights to a lightweight network progressively.

\subsection{Model Optimization and Top-K Recommendation}
For model optimization, we adopt the pairwise Bayesian personalized ranking (BPR) loss~\cite{rendle2009bpr}, that is:
\begin{equation}
	\mathcal{L}_{bpr} = \sum_{(u,i,j)\in \mathcal{R}} \left(-\mathrm{log} \sigma(\mathbf{H}_u^\top \widetilde{\mathbf{H}}_i - \mathbf{H}_u^\top \widetilde{\mathbf{H}}_j)\right).
	\label{eq:loss}
\end{equation}

Additionally, we add regularization penalty on the learned user/item representations.
Our final loss function is:
\begin{equation}
	\mathcal{L} = \mathcal{L}_{bpr} + \lambda \cdot (||\mathbf{H}_u||^2_2 + ||\widetilde{\mathbf{H}}_i||^2_2).
	\label{eq:final_loss}
\end{equation}

To generate item recommendations for a user, we first predict the interaction scores between the user and candidate items. Then, we rank candidate items based on the predicted interaction scores in a descending order, and choose $K$ top-ranked items as recommendations to the user. 
The score for an interaction between user $u$ and item $i$ is calculated as:
\begin{equation}
	r(\mathbf{h}_u, \mathbf{h}_i) = \mathbf{H}_u^\top \widetilde{\mathbf{H}}_i.
\end{equation}
A high score suggests that the user prefers the item.

\subsection{\mmnm{} Algorithm \& Computational Complexity Analysis}

\textbf{PyTorch-style Pseudo-code.}
The pseudo-code of \mmnm{} is in Algorithm 1 \& 2.

\begin{algorithm}
	\caption{PyTorch-style pseudo-code of \mmnm{}.}
	\label{influx}
	\begin{algorithmic}[1]
		\Require A batch of input user-item interactions $\mathbb{B}$, item multimodal feature set $\{\mathbf{X}^m\}$
		\Require $\{f_m\}, f_{sh}, g_{ui}, g_{ii}$ \Comment{$m$ modality-specific encoders, multimodal shared encoder, GCN encoders}
		\For{$m ~ in ~ \mathcal{M}$}\Comment{uni-modal data} 
		\State $\hat{\mathbf{X}}^m = f_m(\mathbf{X}^m)$ \Comment{Eq. (11)} 
        \EndFor
        \State $\hat{\mathbf{X}}^{sh} = f_{sh}\bigl(T(\mathbf{X}^0) \oplus \cdots \bigr)$ \Comment{Eq. (13)} 
        \State $\widetilde{\mathbf{X}} = \hat{\mathbf{X}}^{0}  \oplus \cdots \hat{\mathbf{X}}^{m} \cdots \oplus \hat{\mathbf{X}}^{sh}$ \Comment{Latent item representation}
		\State $({\mathbf{E}}_u, \widetilde{\mathbf{X}}_i) = g_{ui}(\mathbf{E}, \widetilde{\mathbf{X}})$ \Comment{output of GCN, Eq. (15)-(16)} 
        \State $\widetilde{\mathbf{H}}_i = g_{ii}(\widetilde{\mathbf{X}}_i)$ \Comment{GCN on item-item graph, Eq. (17)}
	\State $\mathbf{H}_u = \text{PWC}(\mathbf{E}_u)$ \Comment{PWC in Algorithm 2, Eq. (18)-(20)} 
        \State $\mathcal{L}_{bpr} = \left(-\mathrm{log} \sigma(\mathbf{H}_u^\top \widetilde{\mathbf{H}}_i - \mathbf{H}_u^\top \widetilde{\mathbf{H}}_j)\right)$ \Comment{Eq. (21)} 
        \State $\mathcal{L} = \mathcal{L}_{bpr} + \lambda \cdot (||\mathbf{H}_u||^2_2 + ||\widetilde{\mathbf{H}}_i||^2_2)$ \Comment{Eq. (22)} 
		\State $\mathcal{L}$.backward() \Comment{back-propagate}
	\end{algorithmic}
\end{algorithm}

\begin{algorithm}
	\caption{Pseudo-code of PWC.}
	\label{influxx}
	\begin{algorithmic}[1]
		\Require $B, \mathbf{E}_u$ \Comment{\# of chunks, latent user embedding}
        \Require $\{f_b\}, \theta, \enspace \text{and parameters:}\tau, \eta, n $ \Comment{A set of chunk encoders, tightened target network}
        \Function {PWC:}{}
  		\For{$b ~ in ~ [B]$}\Comment{Each chunk} 
		\State $a_b = f_b(\mathbf{E}^b_u)$ \Comment{weight for unit, Eq. (17)} 
        \State $\eta = \gamma^n$ \Comment{Eq. (18) for lines 4-7.} 
        \State $a_b = \frac{\tau \cdot \eta}{\tau \cdot \eta + 1 - \tau} \cdot a_b + \frac{1 - \tau}{\tau \cdot \eta + 1 - \tau} \cdot \theta_b$
        \State $\tau = \tau \cdot \eta$
        \State $\theta_b = a_b$
        \EndFor
        \State return $\mathbf{H}_u = a_0  \mathbf{E}_u^0 \oplus \cdots  \oplus a_{B-1} \mathbf{E}_u^{B-1})$ \Comment{Eq. (19)}
        \EndFunction 
	\end{algorithmic}
\end{algorithm}

\textbf{Computational Complexity Analysis.}
To ensure consistency, we adopt the notation established in Section \ref{sec:theo}.
The computational cost of \mmnm{} can be decomposed into two components: \textit{the multimodal transformation} and \textit{the graph convolutional operators}.

i) \textit{The multimodal transformation}: 
The projection complexity for multimodal features is $\mathcal{O}(\sum_{m \in \mathcal{M}} |\mathcal{I}|(d_m d_1 + d_1 d))$ on all modalities. 
In addition, \mmnm{} also perform projection on shared multimodal features pre-transformed by 2-D DCT, which has complexity of $\mathcal{O}\left(|\mathcal{I}| (d_1 d + \sum_{m \in \mathcal{M}} d_m d_1) \right)$.

ii) \textit{The graph convolutional operators}: 
While the computational complexity of GCNs applied to user-item and item-item graphs primarily depends on the number of edges, there are slight differences due to the graph structure.

For a user-item graph represented as $\mathcal{G} = (\mathcal{V},\mathcal{E})$, the analytical complexity is $\mathcal{O}(2L_{ui} |\mathcal{E}| d)$, where $L_{ui}$ denotes the number of LightGCN layers, $d$ represents the feature dimension.
The analytical complexity for the item-item graph, denoted by $\mathcal{G}_I = (\mathcal{V}_I,\mathcal{E}_I)$, is $\mathcal{O}(L_{ii} |\mathcal{E}_I| d)$.

To summarize, the overall computational complexity of \mmnm{} is:
$\mathcal{O}\Bigl(2L_{ui} |\mathcal{E}| d + L_{ii} |\mathcal{E}_I| d + |\mathcal{I}| \bigl((|\mathcal{M}|+1) d_1 d + 2\sum_{m \in \mathcal{M}} d_m d_1\bigr) \Bigl)$.

\section{Experiments}
\label{sec:experiments}
To empirically assess the efficacy of our proposed recommendation paradigm, \textit{\mmnm{}}, we conduct a comprehensive experimental evaluation designed to address the following research questions using three real-world datasets.\\
\textbf{RQ1}: To what extent does our proposed method, \mmnm{}, outperform current state-of-the-art recommendation models when assessed using established evaluation metrics and benchmark datasets?\\
{\textbf{RQ2}}: Given that our \mmnm{} framework learns item representations directly from multimodal features, we pose the question: To what degree can our proposed methodology alleviate the issues related to sparse data and demonstrate resilient performance under cold-start scenarios, particularly when there is no interaction information available for specific items?\\
{\textbf{RQ3}}: What are the relative contributions of the individual components comprising our proposed \mmnm{} framework in determining its recommendation performance?

\subsection{Datasets \& Evaluation Protocols}
\label{sec:exp_data}
\begin{table}[bpt]
	\centering	
	\caption{Statistics of the experimental datasets.}
	\begin{tabular}{l r r r r}
		\toprule
		Dataset & \# Users & \# Items & \# Interactions & Sparsity \\
		\midrule
		Baby & 19,445 & 7,050 & 160,792 & 99.88\% \\
		Sports & 35,598 & 18,357 & 296,337 & 99.95\%\\
		Clothing & 39,387 & 23,033 & 278,677 & 99.97\%\\
		Electronics & 192,403 & 63,001 & 1,689,188 & 99.99\%\\
        MicroLens & 98,129 & 17,228 & 705,174 & 99.96\%\\
		\bottomrule
	\end{tabular}
	\vspace{-12pt}
	\label{tab:datasets}
\end{table}

\textit{\textbf{Datasets}.}
Following the experimental setups adopted in prior related works~\cite{LATTICE2021MM, FREEDOM2023MM}, we conduct a comprehensive empirical evaluation of our proposed model and all baselines on five publicly-available benchmark datasets from the Amazon review~\cite{he2016ups} and a short-video platform~\cite{ni2023content}.
Specifically, we use the following categories of Amazon reviews: \textit{i)} \emph{Baby}, \textit{ii)} \emph{Sports and Outdoors}, \textit{iii)} \emph{Clothing, Shoes and Jewelry} and \textit{iv)} \emph{Electronics}. 
For clarity and space consideration, we refer to these three datasets as \emph{Baby}, \emph{Sports},  \emph{Clothing}, and \emph{Elec}, respectively.

\textit{\textbf{Pre-extracted Multimodal Features}.}
To incorporate multimodal item information into recommendation, we leverage the pre-extracted visual and textual features provided by FREEDOM~\cite{MMRec2023MMAsia, FREEDOM2023MM}. That is, we utilize the 4,096-dimensional visual embeddings and the 384-dimensional textual embeddings obtained through their feature extraction pipeline.

\textit{\textbf{MLLM-based Features Extraction}.}
Large Language Models (LLMs) and Multimodal Large Language Models (MLLMs) demonstrate sophisticated capabilities in encoding semantic information from text. Recent research has introduced various text embedding methodologies leveraging these models~\cite{wang2024improving, behnamghader2024llm2vec}. Building upon this foundation, our approach employs MLLMs for image-to-text conversion of item images, followed by the extraction of embeddings from both the generated descriptions and the original textual item descriptions using LLMs.
In our implementation, we leverage Meta's ``Llama-3.2-11B-Vision''~\cite{dubey2024llama} for converting visual content into textual descriptions. 
Subsequently, we employ the ``e5-mistral-7b-instruct'' model~\cite{wang2024improving} to generate text embeddings for the items. 
Through this pipeline, we generate distinct 4,096-dimensional embedding vectors for both the image-derived and text-based item descriptions.

\textit{\textbf{General Evaluation}.}
In this scenario, we adopt the data splitting strategy as delineated in~\cite{LATTICE2021MM, FREEDOM2023MM} for each individual user within the evaluated dataset. 
More specifically, we randomly split each user's historical interactions into three mutually exclusive subsets, with an 8:1:1 ratio for training, validation, and testing, respectively. Given the stipulation  that each user must have a minimum of five interactions within the processed dataset, this splitting approach ensures that at least one sample is available for both validation and testing, while guaranteeing that no fewer than three interactions are utilized for model training.

\textit{\textbf{Cold-Start Evaluation}.}
Following~\cite{LATTICE2021MM}, in cold-start setting, we first randomly select a subset comprising 20\% of the items from the entire item pool within the original dataset. We then divide this held-out set of items into two equal partitions of 10\% each, allocating one partition to serve as the validation set and the other as the test set.
Subsequently, we assign user-item interactions to the training, validation, or testing sets based on the specific item involved in each interaction. Consequently, the items in the validation and test sets are entirely unseen during the model training phase, accurately simulating the challenges posed by cold-start situations where no prior information is available for a subset of items during the recommendation process.

To ensure a fair and consistent comparison with prior works, we adopt the same rigorous evaluation settings employed in existing studies~\cite{LATTICE2021MM, FREEDOM2023MM}. Specifically, we compute the evaluation metrics using the all-ranking protocol and report the average performance of all test users under both $K=10$ and $K=20$ settings for metrics of Recall@$K$ and Normalized Discounted Cumulative Gain (NDCG@$K$), abbreviated as R@$K$ and N@$K$ respectively.

\begin{table*}[bpt]
	\centering	
	\def\arraystretch{0.86}	
        \setlength{\tabcolsep}{2.6pt}
	\caption{Performance comparison of different recommendation methods in terms of Recall and NDCG. `*' denotes that the improvements are statistically significant compared of the best baseline in a paired $t$-test with $p<0.01$. }
	\begin{tabular}{llcccccccccccc}
		\toprule
		\multirow{2}{*}{Dataset} & \multirow{2}{*}{Metric} & MF & LightGCN & VBPR & GRCN & LATTICE & SLMRec & MICRO & BM3 & FREEDOM & LGMRec & \textbf{\mmnm{}} & \textbf{LIRDRec$_\textrm{{MLLM}}$} \\
		\cmidrule(lr){3-4} \cmidrule(lr){5-12} \cmidrule(lr){13-14}
		& & UAI'09 & SIGIR'20 & AAAI'16 & MM'20 & MM'21 & TMM'22 & TKDE'23 & WWW'23 & MM'23 & AAAI'24 & (Ours) & (Ours) \\
		\midrule
            \multirow{4}{*}{Baby} & R@10 & 0.0357 & 0.0479 & 0.0423 & 0.0532 & 0.0540 & 0.0547 & 0.0570 & 0.0564 & {0.0627} & \underline{0.0644} & 0.0693* & \textbf{0.0707*} \\
		& R@20 & 0.0575 & 0.0754 & 0.0663 & 0.0824 & 0.0850 & 0.0810 & 0.0905 & 0.0883 & {0.0992}  & \underline{0.1002} & {0.1048*} & \textbf{0.1068*} \\
		& N@10 & 0.0192 & 0.0257 & 0.0223 & 0.0282 & 0.0292 & 0.0285 & 0.0310 & 0.0301 & {0.0330}  & \underline{0.0349}  & {0.0374*}  & \textbf{0.0382}\\
		& N@20 & 0.0249 & 0.0328 & 0.0284 & 0.0358 & 0.0370 & 0.0357 & 0.0406 & 0.0383 & {0.0424}  & \underline{0.0440} & {0.0465*}  & \textbf{0.0475} \\
		\midrule
		\multirow{4}{*}{Sports} & R@10 & 0.0432 & 0.0569 & 0.0558 & 0.0599 &  0.0620 & 0.0676 & 0.0675 & {0.0656} & {0.0717} & \underline{0.0720} & {0.0785*} & \textbf{0.0815*} \\
		& R@20 & 0.0653 & 0.0864 & 0.0856 & 0.0919 &  0.0953 & 0.1017 & 0.1026 & {0.0980} & \underline{0.1089}  & {0.1068}  & {0.1188*}  & \textbf{0.1216*} \\
		& N@10 & 0.0241 & 0.0311 & 0.0307 & 0.0330 &  0.0335 & 0.0374 & 0.0365 & {0.0355} & {0.0385} &  \underline{0.0390} & {0.0425*}  & \textbf{0.0446*} \\
		& N@20 & 0.0298 & 0.0387 & 0.0384 & 0.0413 &  0.0421 & 0.0462 & 0.0463 & {0.0438} & \underline{0.0481} & {0.0480}  & {0.0529*}  & \textbf{0.0549*} \\
		\midrule
		\multirow{4}{*}{Clothing} & R@10 & 0.0206 & 0.0361 & 0.0281 & 0.0421 &  0.0492 & 0.0540 & 0.0496 & 0.0422 & \underline{0.0629} & {0.0555} & {0.0671*} & \textbf{0.0735*} \\
		& R@20 & 0.0303 & 0.0544 & 0.0415 & 0.0657 &  0.0733 & 0.0810 & 0.0743 & 0.0621 & \underline{0.0941}  & {0.0828} & {0.0998*}  & \textbf{0.1061*} \\
		& N@10 & 0.0114 & 0.0197 & 0.0158 & 0.0224 &  0.0268 & 0.0285 & 0.0264 & 0.0231 & \underline{0.0341}  & {0.0302} & {0.0362*}  &  \textbf{0.0400*} \\
		& N@20 & 0.0138 & 0.0243 & 0.0192 & 0.0284 & 0.0330 & 0.0357 & 0.0332 & 0.0281 & \underline{0.0420}  & {0.0371} & {0.0445*}  & \textbf{0.0483*} \\
            \midrule
            \multirow{4}{*}{Elec} & R@10 & 0.0235 & {0.0363} & 0.0293 &  0.0349 & OOM & 0.0422 & OOM & \underline{0.0437} & 0.0396   & {0.0417}  & {0.0477*} & \textbf{0.0493*}\\
			& R@20 & 0.0367 & 0.0540 & 0.0458 & 0.0529 &  OOM & 0.0630 & OOM & \underline{0.0648} & 0.0601 &   {0.0625} & {0.0709*} & \textbf{0.0736*}\\
			& N@10 & 0.0127 & 0.0204 & 0.0159 & 0.0195 &  OOM & 0.0237 & OOM & \underline{0.0247} & 0.0220  &  {0.0233} & {0.0264*} & \textbf{0.0274*} \\
			& N@20 & 0.0161 & {0.0250} & 0.0202 & 0.0241 &  OOM  & 0.0291 & OOM & \underline{0.0302} & 0.0273  & {0.0287}  & {0.0324*} & \textbf{0.0337*}\\
            \midrule
            \multirow{4}{*}{Microlens} & R@10 & 0.0624 & {0.0720} & 0.0677 & 0.0702 & 0.0726 & \underline{0.0778} & 0.0706 & {0.0606} & 0.0674  & {0.0748}  & {0.0839*} & \textbf{0.0857*} \\
			& R@20 & 0.0959 & 0.1075 & 0.1026 &  0.1070 &  0.1089 & \underline{0.1190} & 0.1101 & {0.0981} & 0.1032 &  {0.1132} & {0.1249*} & \textbf{0.1268*} \\
			& N@10 & 0.0322 & 0.0376 & 0.0351 & 0.0365 &  0.0380 & \underline{0.0405} & 0.0360 & {0.0304} & 0.0345  &  {0.0390} & {0.0440*} & \textbf{0.0453*} \\
			& N@20 & 0.0408 & {0.0467} & 0.0441 & 0.0460 &  0.0473 & \underline{0.0511} & 0.0461 & {0.0400} & 0.0437  & {0.0489}  & {0.0545*}  & \textbf{0.0559*}\\
		\bottomrule	
        \multicolumn{14}{l}{`OOM' signifies an Out Of Memory condition for a Tesla V100 GPU equipped with 32GB of memory.}\\
	\end{tabular}
 \vspace{-10pt}
	\label{tab:perform}	
\end{table*}

\subsection{Baselines}
To rigorously validate the efficacy of our proposed \mmnm{}, we conduct an extensive comparative evaluation against state-of-the-art recommendation models spanning two distinct categories. The first category encompasses two widely-adopted collaborative filtering models that leverage solely user-item interaction data to generate personalized item recommendations for users. 
The second category of baselines includes state-of-the-art multimodal recommendation models that integrate user-item interaction data with auxiliary modalities capturing rich item characteristics.

\noindent\textbf{\textit{i) General models:}} 
\begin{itemize}[leftmargin=*]
\item \textbf{MF}~\cite{rendle2009bpr} learns low-dimensional latent representations for users and items with an objective to minimize the Bayesian Personalized Ranking (BPR) loss within the matrix factorization framework.
\item \textbf{LightGCN}~\cite{he2020lightgcn} is a tailored version of vanilla GCNs, specifically designed for recommendation. It achieves this by eliminating the non-linear activation and feature transformation components, which are typically present in conventional GCN architectures.
\end{itemize}
\textbf{\textit{ii) Multimodal models:}} 
\begin{itemize}[leftmargin=*]
\item \textbf{VBPR}~\cite{VBPR2016AAAI} builds upon the matrix factorization framework by incorporating visual features alongside item ID embeddings to learn enriched item representations. 
\item \textbf{GRCN}~\cite{wei2020graph} further purifies the user-item interactions by identifying and subsequently removing noisy relations. It then learns user and item representations by propagating and aggregating information on the denoised user-item graph using GCNs.
\item \textbf{LATTICE}~\cite{LATTICE2021MM} is designed to dynamically construct latent semantic structures between items for each available data modality. Subsequently, it integrates these modality-specific graphs and executes a process of information propagation and aggregation, utilizing both the user-item interactions and the constructed item-item relationships.
\item \textbf{SLMRec}~\cite{tao2022self} incorporates self-supervised learning to obtain powerful representations for users and items based on the relations underlying the data.
\item \textbf{MICRO}~\cite{zhang2023latent} employs a contrastive learning framework based on LATTICE to align item representations of different modalities with fused multimodal representations, thereby enhancing recommendation accuracy.
\item \textbf{BM3}~\cite{BM32023WWW} perturbs the latent vector representations of items in the model, enabling more effective personalized recommendation even in the absence of negative sampling. 
\item \textbf{FREEDOM}~\cite{FREEDOM2023MM} posits that the dynamic learning of the item-item graph, as implemented in LATTICE, is not requisite. Hence, it constructs the item-item graph from multimodal features and then freezes the constructed graph for effective graph learning.
\item \textbf{LGMRec}~\cite{guo2024lgmrec} initially learns user and item embeddings through local collaborative and modality-specific views, then integrates these embeddings with global representations derived from a hypergraph of users and items.
\end{itemize}

To ensure a fair comparison, we evaluate all baseline models and our proposed model within the same framework as in~\cite{FREEDOM2023MM}. We employ the Adam optimizer for hyperparameter optimization of each model, utilizing grid search to identify their optimal performance on each dataset.
For our proposed \mmnm{}, we follow prior work~\cite{LATTICE2021MM, FREEDOM2023MM} by setting the user preference embedding size to 64 and initializing the parameters with the Xavier method~\cite{glorot2010understanding}. We empirically set the number of GCN layers in the item-item graph ($L_{ii}$) to 1. The remaining hyperparameters of \mmnm{} (number of GCN layers in the user-item graph $L_{ui}$, decay rate $\tau$, decay base $\gamma$, and regularization coefficient $\lambda$) are optimized using grid search across all datasets.

\subsection{Performance Comparison}
\textit{\textbf{General Evaluation of \mmnm{} ({RQ1})}}.
The comparative analysis of recommendation accuracy, quantified in terms of Recall and NDCG, for our proposed model as well as the baseline models is delineated in Table~\ref{tab:perform}. 
In the table, we highlight the global best results on each dataset for each metric in \textbf{boldface} and the best performance for baselines in \underline{underlined}. 
From the table, we have the following observations:

\textit{i) The proposed recommendation paradigm, \mmnm{}, exhibits superior performance compared to all baseline models.} Specifically, \mmnm{} demonstrates a 7.11\% improvement in NDCG@20 performance over the best-performing baseline across all datasets evaluated. When trained with features extracted using MLLMs, \mmnm{}$_\text{MLLM}$ further enhances this improvement to 11.61\%.
Conventional recommendation methods leverage pre-defined item ID embeddings and employ various alignment techniques to enforce congruence between these embeddings and the learned multimodal features. This alignment enables the multimodal features to contribute, albeit to a limited extent, to the learning of item ID representations. Nevertheless, the loss function preferentially optimizes the learning of multimodal features, resulting in suboptimal item ID embeddings, as illustrated in Figure~\ref{fig:intro_loss}.

In contrast, the proposed \mmnm{} dispenses with item ID embeddings entirely, opting instead to directly derive item representations from the inherent information within the multimodal features. Furthermore, it employs a progressive weight copying to modulate the contribution of various uni-modal features, effectively differentiating their importance. This strategy yields item representations that retain rich modality-specific information crucial for recommendation, while simultaneously mitigating the inclusion of irrelevant details by assigning lower weights to less informative features.

\textit{ii) Recommendation models can potentially benefit from the multifaceted approach of incorporating multimodal features.}
Generally, both traditional models (\eg MF) and graph-based models (\eg LightGCN) can benefit from this enriched information space. However, models such as GRCN, FREEDOM, and \mmnm{} take a more nuanced approach, leveraging multimodal features on multiple levels to obtain better recommendation performance. Specifically, GRCN utilizes multimodal features to: \textit{a)} weight the importance of edges within the user-item interaction graph, and \textit{b)} fuse the features into the final user and item representations for recommendation. As a result, GRCN achieves significant improvements over MMGCN. Similarly, FREEDOM and \mmnm{} employ multimodal features to construct an item-item interaction graph and enable the propagation of latent item representations within this graph. These models further integrate multimodal features directly into the loss function, facilitating the effective learning of user and item representations during training. This multifaceted integration demonstrably enhances recommendation performance, as evidenced by the superior accuracy achieved by both FREEDOM and \mmnm{}.

Given that \mmnm{} directly learns item representations from multimodal features, the quality of these features significantly influences its performance. Leveraging the rich semantic information and expressive power of MLLMs and LLMs, features extracted from these models demonstrate superior recommendation performance compared to those obtained using conventional methods such as CNNs. Specifically, \mmnm{}$_\text{MLLM}$ achieves a 4.21\% improvement in NDCG@20 performance over the baseline \mmnm{} across all evaluated datasets.

\begin{figure}[bpt]
  \centering
  \includegraphics[width=\linewidth, trim={0 10 10 10},clip]{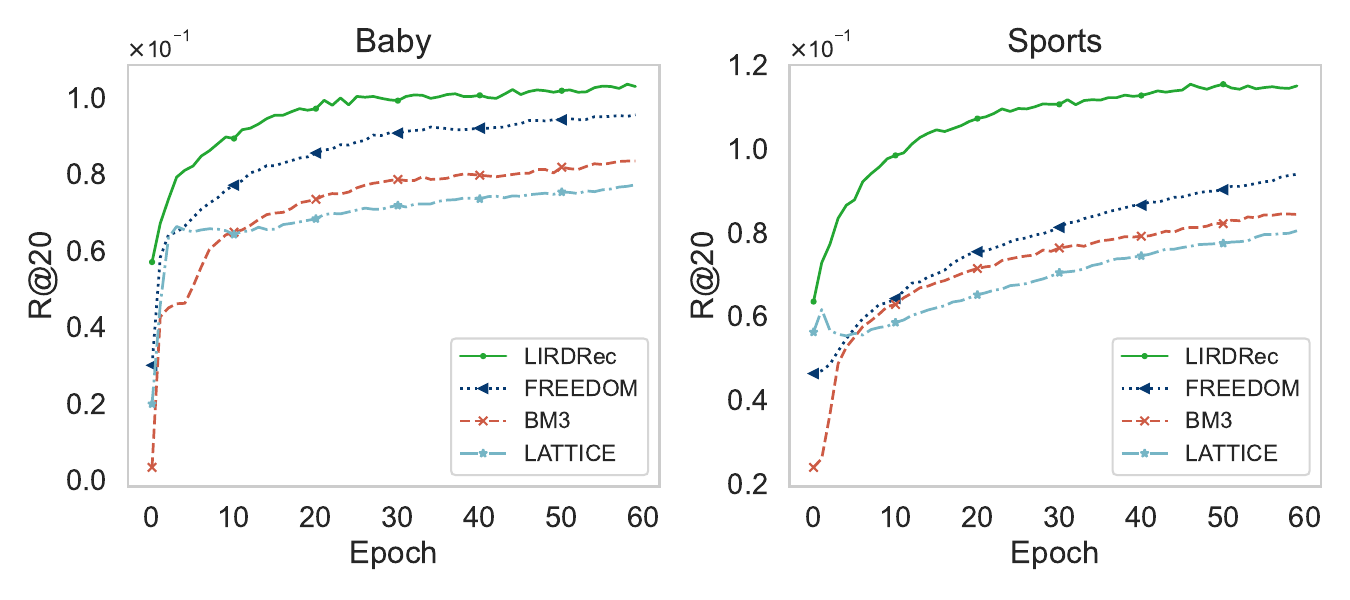}
  \vspace{-10pt}
  \caption[]{The proposed \mmnm{} can quickly boost startup recommendation performance for both datasets.}
  \label{fig:startup}
  \vspace{-12pt}
\end{figure}

\textit{\textbf{Quick Startup of \mmnm{}}}.
As detailed in Section~\ref{sec:model}, the benefits of directly learning from multimodal features are manifold. A key advantage is the capacity to rapidly acquire informative representations by attending to multimodal information.

In order to further investigate this, we scrutinize the training logs of four models: LATTICE, BM3, FREEDOM, and \mmnm{}. The performance of these models, specifically the Recall@20, is plotted over the first 60 training epochs as shown in Fig.~\ref{fig:startup}.
From the figure, we observe a dramatic improvement in recommendation accuracy for \mmnm{} after the initial training epoch (0-indexed). Notably, \mmnm{} exhibits a performance gain of between two and three times that of both FREEDOM and BM3.

\textit{\textbf{Cold-Start Evaluation of \mmnm{} ({RQ2})}}.
Multimodal recommendation models leverage additional information beyond user-item interactions, contributing to alleviate the data sparsity issue. 
Hence, we further evaluate a set of multimodal model under a cold-start scenario.
Fig.~\ref{fig:coldstart} presents a comparative analysis of recommendation performance for our proposed \mmnm{} and three representative models.
As general recommendation models (\ie MF and LightGCN) struggle to recommend relevant items for cold-start users due to their heavy reliance on initial user/item embedding values, their results are omitted from the figure for clarity.
The figure elucidates the subsequent facets:

\begin{figure}[bpt]
  \centering
  \includegraphics[width=0.88\linewidth]{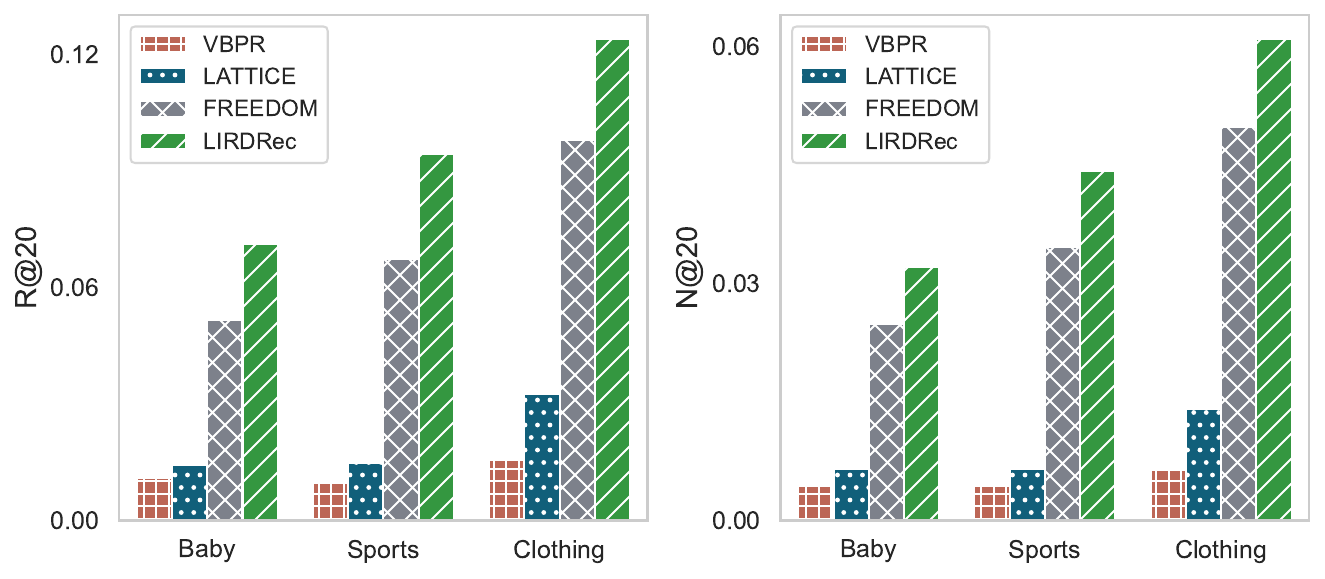}
  \vspace{-6pt}
  \caption[]{Performance of \mmnm{} compared with representative baselines under cold-start settings.}
  \vspace{-2pt}
  \label{fig:coldstart}
\end{figure}

\textit{Firstly}, incorporating multimodal features into the training loss function can enhance the robustness of recommendation models in cold-start scenarios. This is supported by the following two observations. \textit{i) VBPR:} This model concatenates multimodal features with item IDs for item representation learning. While LATTICE solely utilizes multimodal features to construct the item-item graph, VBPR achieves competitive performance on smaller datasets like ``Baby'' and ``Sports''. This suggests that directly incorporating multimodal features in the loss function might be beneficial.
\textit{ii) FREEDOM and \mmnm{}:} Both these models integrate multimodal features into item representation learning. Consequently, they significantly outperform LATTICE across all datasets, further supporting the effectiveness of this approach.
\textit{Secondly}, GCNs have the potential to propagate information and gradients to unseen items (cold-start items) during training, even if they have not been observed in user-item interactions. This can alleviate the cold-start problem, particularly when the user-item graph is large and sparsely connected. Partial validation for this can be observed on ``Clothing'' dataset in Fig.~\ref{fig:coldstart}.

\begin{table}[bpt]
	\centering	
	\def\arraystretch{0.85}	
        \setlength{\tabcolsep}{4.6pt}
	\caption{Performance comparison of \mmnm{} variants under different component ablation settings.}
	\vspace{-4pt}
	\begin{tabular}{llccc}
		\toprule
		{Dataset} & {Metric} & \mmnm{}\textsubscript{w/o PWC} & \mmnm{}\textsubscript{w/o DCT}& \textbf{\mmnm{}} \\
		\midrule
		\multirow{2}{*}{Baby} & R@20 & 0.0893 & 0.1019 & 0.1048\\
		& N@20 & 0.0403 & 0.0449 & 0.0465 \\
		\midrule
		\multirow{2}{*}{Sports} & R@20 & 0.0994 & 0.1137 & 0.1188 \\
		& N@20 & 0.0452 & 0.0516 & 0.0529 \\
		\midrule
		\multirow{2}{*}{Clothing} & R@20 & 0.0857 & 0.0959 & 0.0998 \\
		& N@20 & 0.0387 & 0.0433 & 0.0445 \\
		\bottomrule
	\end{tabular}
	\vspace{-12pt}
	\label{tab:perform_com}	
\end{table}

\begin{table}[bpt]
	\centering
	\def\arraystretch{0.85}	
	\caption{Performance comparison of \mmnm{} variants under different uni-modal features.}
	\begin{tabular}{llccc}
		\toprule
		{Dataset} & {Metric} & \mmnm{}\textsubscript{w/o V} & \mmnm{}\textsubscript{w/o T}& \textbf{\mmnm{}} \\
		\midrule
		\multirow{2}{*}{Baby} & R@20 & 0.1033 & 0.0872 & 0.1048\\
		& N@20 & 0.0446 & 0.0384 & 0.0465 \\
		\midrule
		\multirow{2}{*}{Sports} & R@20 & 0.1156 & 0.0977 & 0.1188 \\
		& N@20 & 0.0518 & 0.0435 & 0.0529 \\
		\midrule
		\multirow{2}{*}{Clothing} & R@20 & 0.0978 & 0.0738 & 0.0998 \\
		& N@20 & 0.0438 & 0.0331 & 0.0445 \\
		\bottomrule
	\end{tabular}
	\vspace{-12pt}
	\label{tab:perform_mm}	
\end{table}

\subsection{Ablation Study (\textbf{RQ3})}
This section presents an ablation study to assess the individual contributions of each component within \mmnm{} and the impact of its input multimodal features on recommendation performance.

Accordingly, we design the following variants of \mmnm{}: \emph{i)} \textbf{\mmnm{}\textsubscript{w/o PWC}} employs a simplified multimodal feature fusion approach in \mmnm{}. It forgoes the PWC strategy, resulting in the concatenation of all multimodal features.
\emph{ii)}  \textbf{\mmnm{}\textsubscript{w/o DCT}} excludes the shared multimodal feature derived from the 2-D DCT transformation and DNN, feeding \mmnm{} with only modality-specific features.
\emph{iii)}  \textbf{\mmnm{}\textsubscript{w/o V}} learns user/item representations and conducts recommendation solely on item features excluding \textit{visual} information.
\emph{iv)}  \textbf{\mmnm{}\textsubscript{w/o T}}  learns user/item representations and conducts recommendation solely on item features excluding \textit{textual} information.

\textit{\textbf{Component Ablation}}.
Table~\ref{tab:perform_com} summarizes the recommendation performance of \mmnm{} under different component ablation settings.
Analysis of Table~\ref{tab:perform_com} reveals that the PWC component has the most significant impact on \mmnm{}'s recommendation performance. Despite this, \mmnm{} remains competitive with most baseline models even without PWC, highlighting the effectiveness of the proposed recommendation paradigm. The multimodal feature transformation contributes less than PWC, but it still improves \mmnm{}'s NDCG@20 score by 2.95\%.

\textit{\textbf{Multimodal Feature Ablation}}.
To analyze the relative importance of visual and textual features in our multimodal model, we conducted a feature ablation study.
Specifically, we trained separate models where either the visual or textual input channel was removed entirely. 
The performance of these models was then compared to the full multimodal model on a held-out test set in Table~\ref{tab:perform_mm}. This controlled setting allowed us to isolate the contribution of each modality to the recommendation accuracy. Our findings first revealed that the model (\mmnm{}\textsubscript{w/o V}) trained solely on textual features achieved a significantly higher recommendation performance compared to the model (\mmnm{}\textsubscript{w/o T}) trained only on visual features. This suggests that textual information plays a more prominent role in obtaining accurate recommendations within this specific domain~\cite{FREEDOM2023MM}.
Comparative analysis of performance metrics in Tables \ref{tab:perform_mm} and \ref{tab:perform_com} reveals an unexpected observation: the model denoted as \mmnm{}\textsubscript{w/o V}, trained solely on textual features, outperforms the model \mmnm{}\textsubscript{w/o DCT}, which utilizes both features but lacks the DCT transformation and the shared representation in Eq.~\eqref{eq:sh}. This counterintuitive finding underscores the critical role of the 2-D DCT component in effectively integrating multimodal information within the \mmnm{} architecture 
It is also worth noting that \mmnm{}\textsubscript{w/o V}, exclusively equipped with textual features, outperforms all baselines equipped with both features.

\subsection{Hyperparameter Sensitivity Study}
To gain deeper insights into the behavior of the proposed \mmnm{} method, we conduct additional experiments on its components and hyperparameters.

\begin{figure}[bpt]
  \centering
  \includegraphics[width=0.9\linewidth, trim={10 10 10 10},clip]{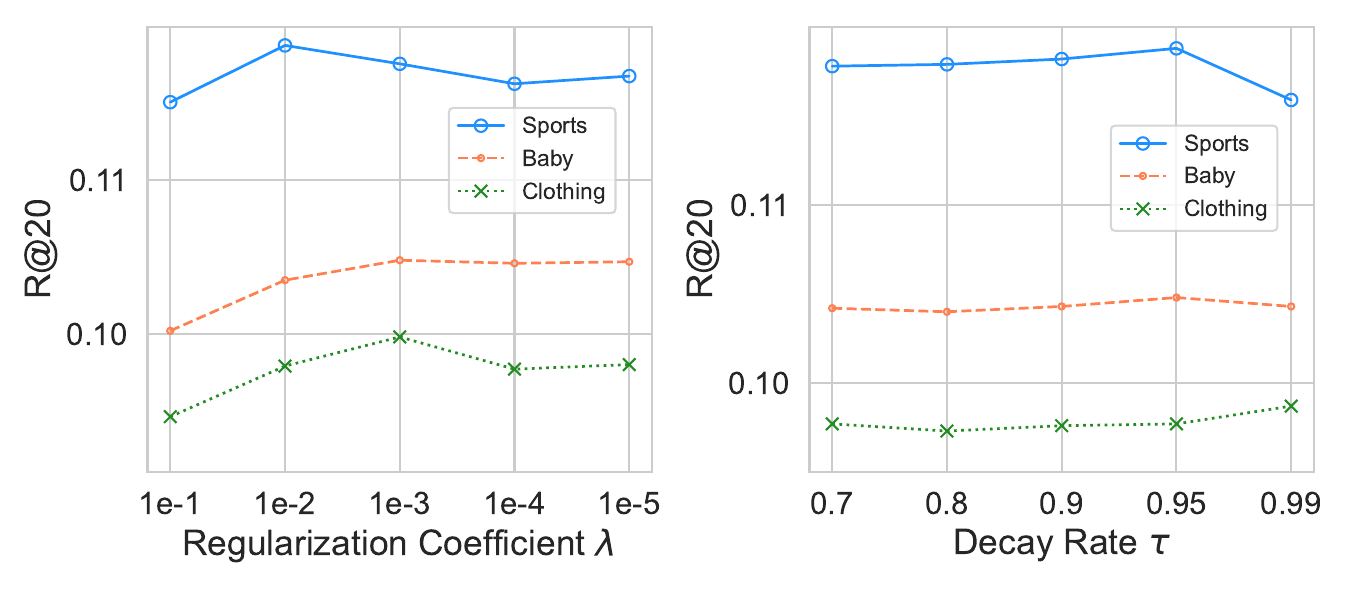}
  \caption[]{Performance of \mmnm{} with regard to different regularization coefficient $\lambda$ and decay rate $\tau$.}
  \label{fig:hyper}
  \vspace{-8pt}
\end{figure}

To assess the impact of hyperparameter settings on model performance, we conducted a sensitivity analysis focusing on the regularization coefficient ($\lambda$) and the decay rate $\tau$ in Fig.~\ref{fig:hyper}. The regularization coefficient was varied across set: \{1$e-$1, 1$e-$2, 1$e-$3, 1$e-$4, 1$e-$5\}, while the decay rate explored a wider range \{0.7, 0.8, 0.9, 0.95, 0.99\}. Our findings revealed that the model's performance was less sensitive to the decay rate than the regularization coefficient. This suggests that the model is more robust to variations in the decay rate, which controls the weight update during training. However, for the evaluated datasets, a larger decay rate ($\ge0.9$) usually yielded superior performance. In contrast, stronger regularization (higher $\lambda$) appears to impede the model's ability to learn effective representations from the multimodal data.

We investigated the effect of the number of layers in the user-item graph on the performance of \mmnm{}. The results are visualized in Fig.~\ref{fig:layers}. The figure suggests that \mmnm{} with GCNs can benefit from a higher number of layers when applied to larger user-item graphs, as observed in the Sports and Clothing datasets.

\begin{figure}[bpt]
  \centering
  \includegraphics[width=0.9\linewidth, trim={10 10 10 10},clip]{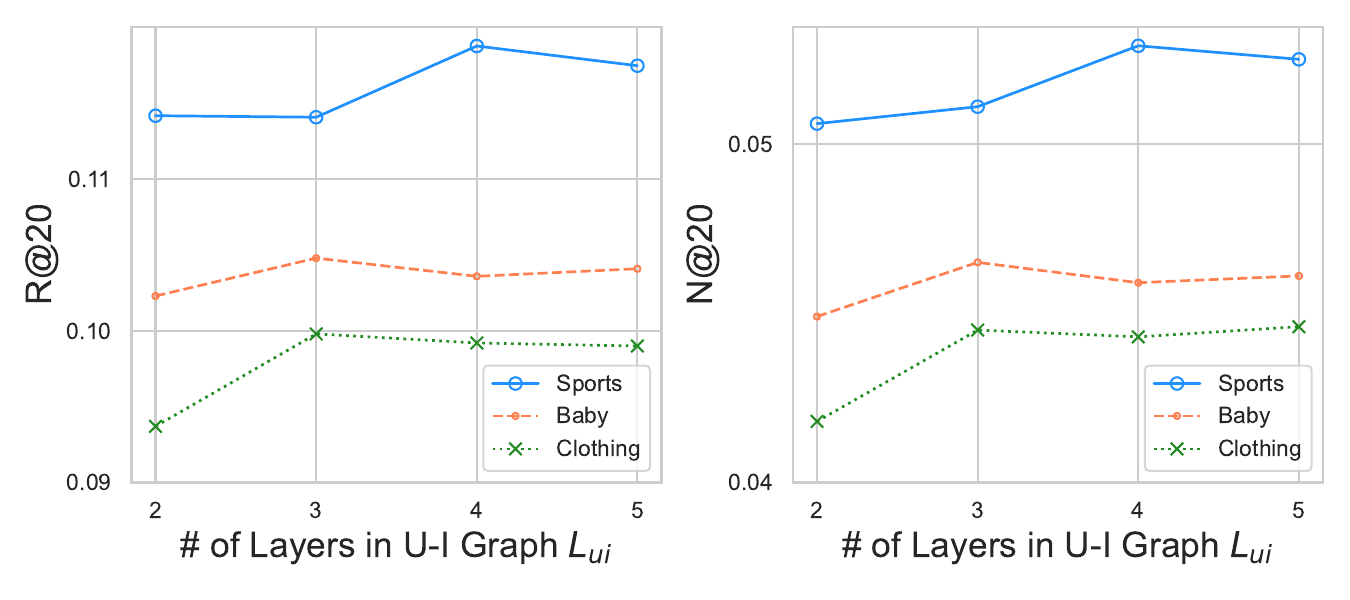}
  \caption[]{Performance of \mmnm{} \emph{w.r.t} varied layers in user-item graph.}
  \vspace{-12pt}
  \label{fig:layers}
\end{figure}

\section{Related Work}
\label{sec:relatedwork}
\subsection{Multimodal Recommender Systems}
Multimodal recommender systems augment the capabilities of traditional collaborative filtering (CF)~\cite{he2020lightgcn,zhou2023selfcf,zhou2023layer} by incorporating a more extensive spectrum of user and item information~\cite{VBPR2016AAAI, LATTICE2021MM, FREEDOM2023MM,zhang2024multiPRICAI,zhang2024multiCIKM}. This involves the assimilation of exogenous data modalities, including but not limited to textual descriptions and visual features, which transcend the boundaries of user-item interactions. As a result, these multimodal recommender systems demonstrably achieve superior performance, especially in situations with data sparsity or the presence of cold-start items~\cite{LATTICE2021MM,xu2021multi,zhang2024multimodal}. Within this domain, an extensive assortment of methodologies has been devised for the efficacious assimilation of multimodal data into modern recommendation architectures.

Early integration methods for multimodal information primarily focused on feature-level fusion. These approaches include the concatenation of multimodal features with ID embeddings~\cite{VBPR2016AAAI, liu2017deepstyle} and attention-based weighting of item features within the embedding space~\cite{chen2017attentive, chen2019personalized}. While these models demonstrate significant performance improvements in recommendation accuracy, they lack the ability to capture higher-order relationships between users and items (or items themselves).
Motivated by the success of GCNs and their inherent suitability for representing user-item interaction graphs, researchers have explored their integration with multimodal recommender systems. MMGCN~\cite{wei2019mmgcn}, for instance, constructs modality-specific user-item bipartite graphs and allows for the joint propagation of user/item ID embeddings and multimodal features. However, this approach fails to differentiate the importance of features from different modalities and effectively align ID embeddings with multimodal representations, resulting in suboptimal performance. To address this, GRCN~\cite{wei2020graph} introduces a refined message passing process within the GCN framework, dynamically weighting interaction edges between users and items based on multimodal information. FREEDOM~\cite{FREEDOM2023MM} builds upon this concept by further eliminating irrelevant edges through a degree-sensitive pruning method.

While multimodal information primarily resides with items, propagating it through the original user-item interaction graph can introduce significant noise, particularly over longer distances. To address this challenge, SGFD~\cite{liu2023semantic} adopts a feature-level approach, employing a semantic-guided teacher-student framework to perform knowledge distillation on the multimodal features. The distilled knowledge can then be leveraged by subsequent multimodal recommendation models, such as GRCN~\cite{wei2020graph} or BM3~\cite{BM32023WWW}, for effective recommendation.
An alternative research trajectory has been investigating the augmentation of prevailing multimodal models through the integration of supplementary item-item graphs~\cite{LATTICE2021MM, MGCN2023MM, FREEDOM2023MM}.
LATTICE~\cite{LATTICE2021MM} exemplifies this approach by explicitly constructing item-item relationship graphs for each modality and subsequently fusing them into a latent item-item representation. This latent graph is then dynamically updated using projected multimodal features obtained from MLPs, leading to state-of-the-art recommendation accuracy.
In contrast to LATTICE, MGCN~\cite{MGCN2023MM} adopts a static item-item graph, performing information propagation alongside each modality for the items. This method leverages an attention mechanism to differentiate the contribution of features from different modalities, leading to improved performance. 
Instead of relying on the item-item graph, MM-FRec~\cite{song2023mm} utilizes other types of graphs, such as the user-item-attribute tripartite graph and the user-image bipartite graph, to learn user and item representations.
Another line of work employs self-supervised learning paradigms for the integration of multimodal data within the structure of the recommendation model~\cite{tao2022self, BM32023WWW}. These methodologies formulate various contrastive learning objectives between the multimodal feature representations and user/item embeddings, thereby efficaciously catalyzing their learning process.

In summary, existing multimodal recommender systems leverage item ID embeddings as a core component. These approaches primarily focus on either incorporating multimodal features into item ID embeddings through alignment or fusion techniques~\cite{VBPR2016AAAI, BM32023WWW}, or exploiting latent relationships within item multimodal features to enhance the representation learning process~\cite{FREEDOM2023MM}. While acknowledging the importance of item ID embeddings in recommendation, our proposed method, \mmnm{}, departs from this paradigm. \mmnm{} eliminates the need for item ID embeddings altogether, instead focusing on directly learning user and item representations centered around the multimodal information. 

\vspace{-5pt}
\section{Conclusion}
In this paper, we shed light on the multimodal feature learning challenge that is inherent in state-of-the-art multimodal recommendation models, substantiated by empirical experiments. On this foundation, we propose  \mmnm{}, a model that exclusively and directly learns item representations from multimodal features. This methodology dispenses with the need for item ID embeddings, a convention prevalent in almost all current multimodal recommender systems. \mmnm{} directly gleans informative item representations from their respective multimodal features and refines these representations in conjunction with users via collaborative filtering. Furthermore, to promote the exchange of multimodal information, we construct a multimodal transformation mechanism employing 2-D DCT and DNNs. Both modality-specific and shared information are subsequently incorporated into a progressive weight copying module to ascertain the relevance of the information. Comprehensive experiments corroborate the superior performance of \mmnm{} within the domain of multimodal recommendations. 
As future research directions, we propose extending \mmnm{} to sequential recommendation scenarios. Specifically, we aim to investigate how the progressive weight copying techniques developed in \mmnm{} can enhance the performance of sequential recommender systems by incorporating temporal dynamics and multimodal feature interactions. This integration presents opportunities to explore the temporal evolution of user preferences across multiple modalities while leveraging the efficient progressive weighting mechanisms established in our framework.

 {\appendix[Case Study: Text Generation Using MLLMs]
 We demonstrate the application of MLLM (\ie Llama 3.2 Vision)for generating text based on the cover image of a product within the Clothing dataset.
 For each product cover image, denoted by its unique identifier `ASIN', we formulate structured prompts to guide MLLMs in generating comprehensive descriptions with respect to predefined product attributes. 
 Through extensive pre-training on vision-language datasets, MLLMs exhibit the capacity to discern product features that may be overlooked in original item descriptions but are potentially of significant value to users. This capacity to extract implicit visual attributes effectively augments existing product metadata, proving particularly beneficial for products with limited or absent descriptive information. 
 For instance, item `B000051SEN', illustrated in Fig.~\ref{fig:mllmstudy}, lacks textual descriptions within the original dataset. However, MLLMs effectively generate descriptive text, accurately capturing the item's salient features (right panel of Fig.~\ref{fig:mllmstudy}).

 \begin{figure*}[ht]
   \centering
   \includegraphics[width=0.82\linewidth, trim={16 115 16 130},clip]{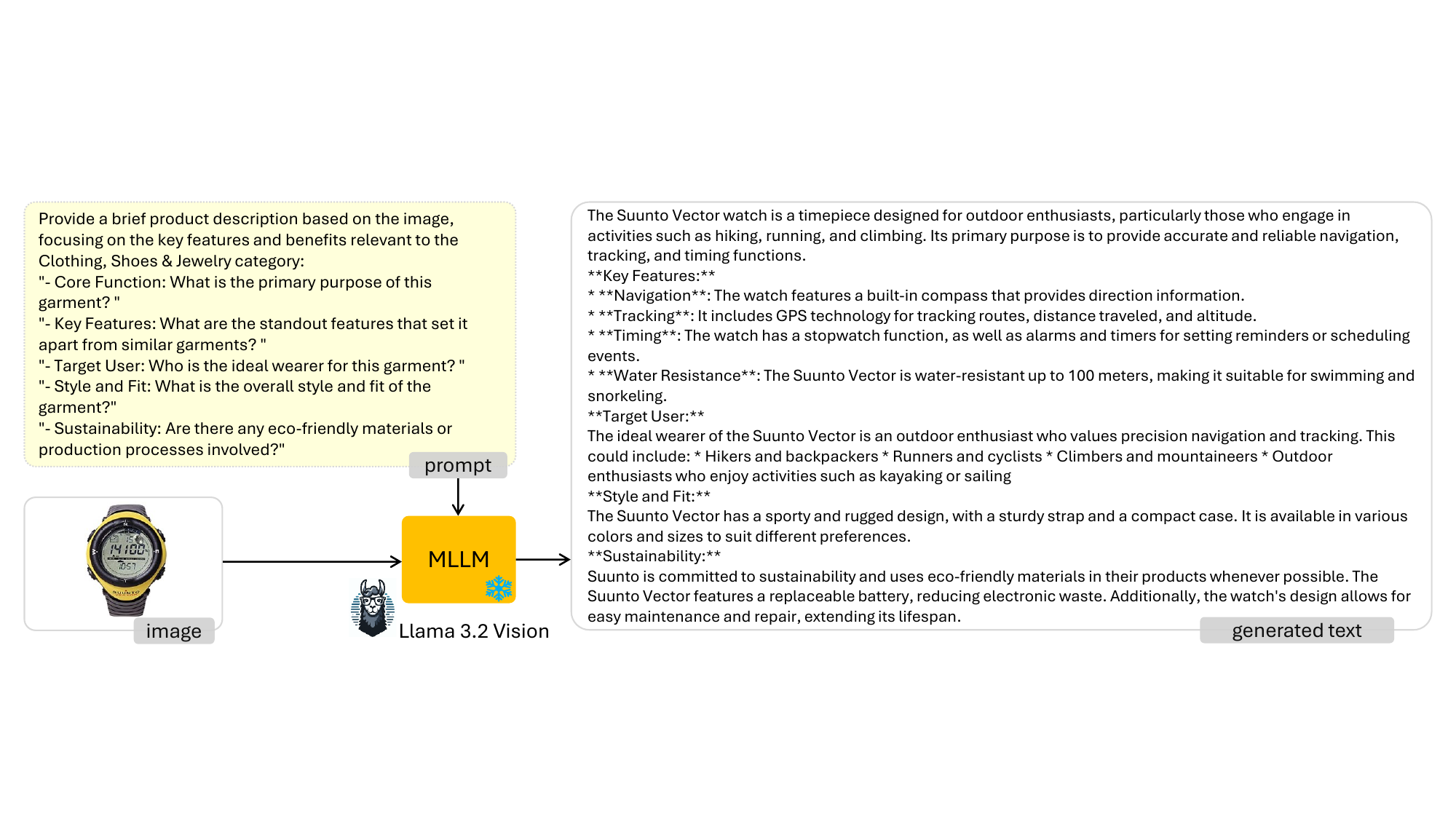}
   \caption[]{Case study on MLLM-based text generation.}
     \vspace{-16pt}
   \label{fig:mllmstudy}
 \end{figure*}
 }

\bibliographystyle{IEEEtran}
\bibliography{IEEEabrv,main}

\end{document}